\pdfoutput=1
\RequirePackage{ifpdf}
\ifpdf
\documentclass[pdftex]{sigma}
\else
\documentclass{sigma}
\fi

\usepackage{bm}

\def\z{{\bm z}}
\def\w{{\bm w}}

\DeclareMathOperator{\card}{card}

\makeatletter
\DeclareRobustCommand\widecheck[1]{{\mathpalette\@widecheck{#1}}}
\def\@widecheck#1#2{
 \setbox\z@\hbox{\m@th$#1#2$}
 \setbox\tw@\hbox{\m@th$#1
 \widehat{
 \vrule\@width\z@\@height\ht\z@
 \vrule\@height\z@\@width\wd\z@}$}
 \dp\tw@-\ht\z@
 \@tempdima\ht\z@ \advance\@tempdima2\ht\tw@ \divide\@tempdima\thr@@
 \setbox\tw@\hbox{%
 \raise\@tempdima\hbox{\scalebox{1}[-1]{\lower\@tempdima\box
\tw@}}}
 {\ooalign{\box\tw@ \cr \box\z@}}}

\numberwithin{equation}{section}

\newtheorem{Theorem}{Theorem}[section]
\newtheorem{Corollary}[Theorem]{Corollary}
\newtheorem{Lemma}[Theorem]{Lemma}
\newtheorem{Proposition}[Theorem]{Proposition}
{ \theoremstyle{definition}
\newtheorem{Definition}[Theorem]{Definition}}

\begin{document}
\allowdisplaybreaks

\newcommand{\arXivNumber}{1711.04893}

\renewcommand{\PaperNumber}{049}

\FirstPageHeading

\ShortArticleName{Jacobi--Trudi Identity in Super Chern--Simons Matrix Model}
\ArticleName{Jacobi--Trudi Identity in Super Chern--Simons\\ Matrix Model}

\Author{Tomohiro FURUKAWA and Sanefumi MORIYAMA}
\AuthorNameForHeading{T.~Furukawa and S.~Moriyama}
\Address{Department of Physics, Osaka City University, Osaka 558-8585, Japan}
\Email{\href{mailto:furukawa@sci.osaka-cu.ac.jp}{furukawa@sci.osaka-cu.ac.jp}, \href{mailto:moriyama@sci.osaka-cu.ac.jp}{moriyama@sci.osaka-cu.ac.jp}}

\ArticleDates{Received January 19, 2018, in final form May 10, 2018; Published online May 18, 2018}

\Abstract{It was proved by Macdonald that the Giambelli identity holds if we define the Schur functions using the Jacobi--Trudi identity. Previously for the super Chern--Simons matrix model (the spherical one-point function of the superconformal Chern--Simons theo\-ry describing the worldvolume of the M2-branes) the Giambelli identity was proved from a~shifted version of it. With the same shifted Giambelli identity we can further prove the Jacobi--Trudi identity, which strongly suggests an integrable structure for this matrix model.}


\Classification{05E05; 37K10}

\section{Introduction}

More than one hundred years ago, two remarkable identities for the Schur polynomial were found. One of them is the Giambelli identity~\cite{G}, stating that the Schur polynomial in any representation is expressed as a determinant of the Schur polynomials in the hook representation. Another is the Jacobi--Trudi identity~\cite{J}, stating that the Schur polynomial is expressed as a~determinant of the complete symmetric polynomials.
It is interesting to see how these two identities are related to each other.

In \cite{M9,Macdonald} in discussing generalizations of the Schur polynomial, in the final generalization (known as {\it the ninth variation} \cite{M9}), which combines many generalizations given previously, Macdonald regards the Jacobi--Trudi identity as the definition of the Schur function in an arbitrary representation out of the complete symmetric functions. Namely, we prepare a set of functions as the complete symmetric functions and define the generalized Schur function in any representation out of them using the Jacobi--Trudi identity (possibly along with an automorphism). Then, it was found that the Giambelli identity follows simply within this setup.

\looseness=-1
It was found \cite{BOS,Tierz} that some one-point functions of the Schur polynomial in matrix models also enjoy the Giambelli identity. Also, following \cite{shifted} which rewrites the one-point function in the ABJM matrix model \cite{ABJM} into an expression which is reminiscent of the Giambelli identity with a shift, in \cite{GC} we were able to prove the original Giambelli identity for the ABJM matrix model.

From the viewpoint of the Macdonald's ninth variation of the Schur function \cite{M9}, it is natural to ask whether the Giambelli identity which is satisfied by the ABJM matrix model is lifted to the Jacobi--Trudi identity, and if yes, what the automorphism is.

In this paper, we shall give a positive answer to this question. In the ABJM matrix model there is a natural choice of the automorphism: the shift in the fractional-brane background. Then, the Jacobi--Trudi identity strongly suggests that the ABJM matrix model has the structure of the integrable hierarchy. Before going into the proof of our statement, we shall make the above motivation clearer by preparing some mathematical concepts and explaining the physical backgrounds.

\subsection{Mathematical preparation}

Let us start with some preparations of mathematical backgrounds.\footnote{We mostly follow the terminology and the notation given in \cite{Macdonald} except for the following points. We use the uppercase characters for the variables $L$ and $R$ which are later related to the matrix size in the determinant formula. We denote the transpose of the Young diagram $\lambda$ by $\lambda^\text{T}$ instead of $\lambda'$ and reserve the primes for the concept related to shifting the diagonal line.} There are a few major notations to express the Young diagram $\lambda$
\begin{gather*}
\lambda=[\lambda_1,\lambda_2,\dots,\lambda_L]=(\alpha_1,\alpha_2,\dots,\alpha_R\,|\,\beta_1,\beta_2,\dots,\beta_R).
\end{gather*}
The first one is the standard one listing all of the box numbers $\lambda_i$ in the $i$-th row (with $L=\card\{i|\lambda_i>0\}$). The second one is the Frobenius notation obtained by counting the horizontal boxes and the vertical boxes from the main diagonal as arm lengths and leg lengths
\begin{gather}
\alpha_i=\lambda_i-i,\qquad \beta_j=\big(\lambda^\text{T}\big)_j-j,\qquad R=\card\{i\,|\,\alpha_i\ge 0\}=\card\{j\,|\,\beta_j\ge 0\}. \label{frobenius}
\end{gather}
Finally, we propose a notation which is not very common. Let $M$ be an integer. Then, the last one is again the Frobenius notation but obtained by shifting the diagonal line by $M$ to the right. Namely, if $M\ge 0$ we define
\begin{gather}
\alpha'_i=\lambda_i-i-M,\qquad \beta'_j=\big(\lambda^\text{T}\big)_j-j+M,\nonumber\\
R'=\card\{i\,|\,\alpha'_i\ge 0\}=\card\{j\,|\,\beta'_j\ge 0\}-M,\label{armleg}
\end{gather}
resulting in $M$ more leg lengths than arm lengths,
\begin{gather}
\lambda=\big(\alpha'_1,\alpha'_2,\dots,\alpha'_{R'}\,|\,\beta'_1,\beta'_2,\dots,\beta'_{M+R'}\big),
\end{gather}
while if $M=-|M|\le 0$ we define\footnote{Considering the continuation from $M\ge 0$, it may be more convenient to define $R'=\card\{i\,|\,\alpha'_i\ge 0\}=\card\{j\,|\,\beta'_j\ge 0\}+|M|$ for $M=-|M|\le 0$.
We adopt this definition because it is more intuitive in our proof.}
\begin{gather}
\alpha'_i=\lambda_i-i+|M|,\qquad \beta'_j=\big(\lambda^\text{T}\big)_j-j-|M|,\nonumber\\
R'=\card\{i\,|\,\alpha'_i\ge 0\}-|M|=\card\{j\,|\,\beta'_j\ge 0\}, \label{armlegnegM}
\end{gather}
resulting in $|M|$ more arm lengths than leg lengths
\begin{gather}
\lambda=\big(\alpha'_1,\alpha'_2,\dots,\alpha'_{|M|+R'}\,|\,\beta'_1,\beta'_2,\dots,\beta'_{R'}\big),
\end{gather}
To equate the numbers of the arm lengths and the leg lengths, we often prepare the auxiliary arm lengths
\begin{gather}
\widetilde\alpha'_i=i-M-1,\qquad 1\le i\le M, \label{aux}
\end{gather}
for $M\ge 0$ and the auxiliary leg lengths
\begin{align}
\widetilde\beta'_j=j-|M|-1,\qquad 1\le j\le|M|, \label{auxnegM}
\end{align}
for $M=-|M|\le 0$. We call the last notation the $M$-shifted Frobenius notation. See Fig.~\ref{young} for an example.
\begin{figure}[!ht]
\centering\includegraphics[scale=0.45,angle=-90]{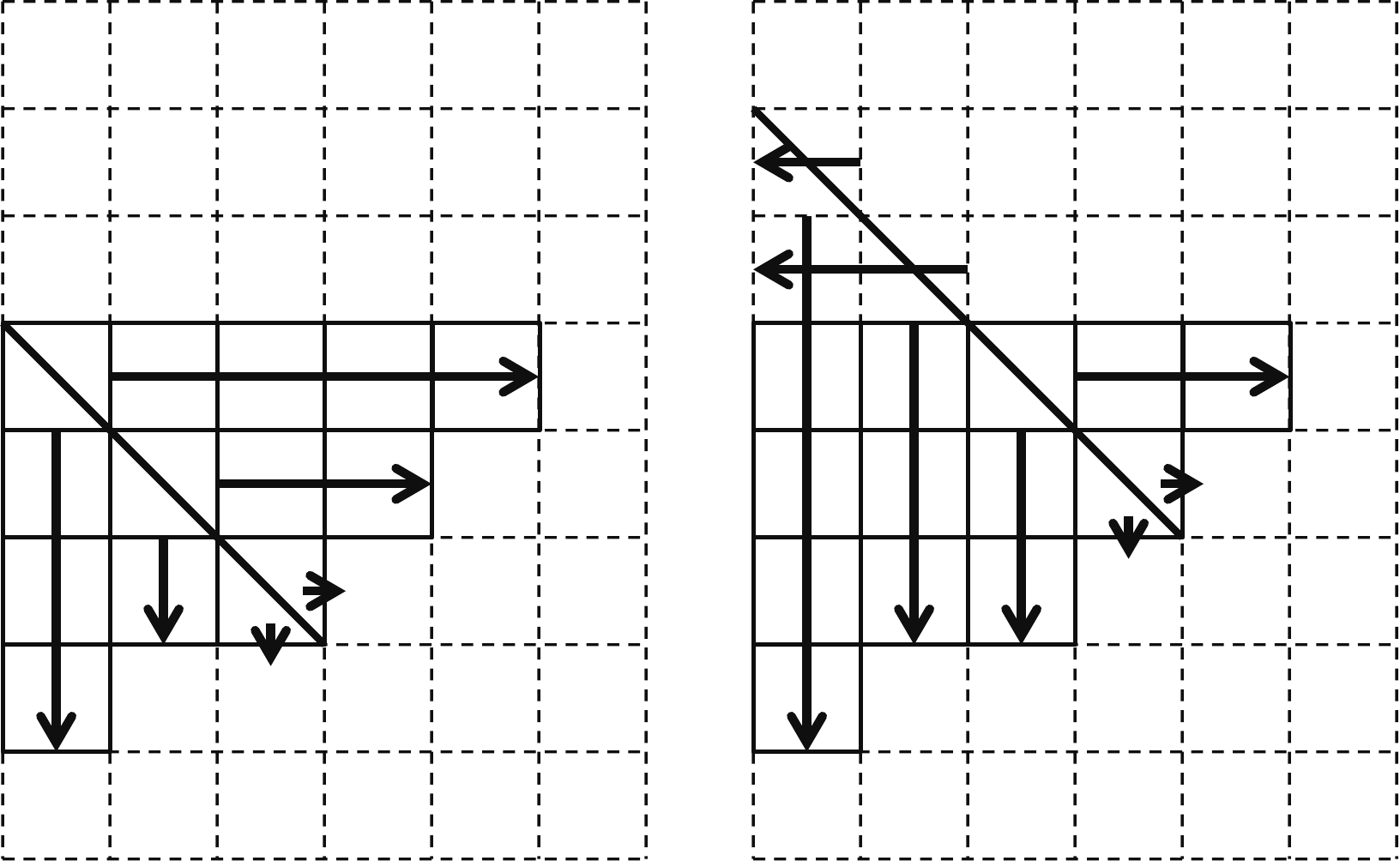}
\caption{The Frobenius notation (left) and the shifted Frobenius notation with $M=2$ (right) for $\lambda=[5,4,3,1]$. The Frobenius notation is obtained by counting the horizontal and vertical boxes from the diagonal, $\lambda=(4,2,0\,|\,3,1,0)$, while the shifted Frobenius notation is obtained similarly with the diagonal shifted by $M=2$ to the right, $\lambda=(2,0\,|\,5,3,2,0)$.}\label{young}
\end{figure}

We consider\footnote{The typical examples of the functions $S_\lambda$, $S^M_\lambda$ are the Schur polynomial and the super Schur polynomial, though in general (and in our application in the next subsection) these functions do not have to be symmetric polynomials.} a set of functions $S_\lambda$ labelled by the Young diagram $\lambda$, which are normalized by $S_\varnothing=1$, or a set of functions $S^M_\lambda$ labelled by both the Young diagram $\lambda$ and the integer $M$, which are normalized by $S^0_\varnothing=1$. Otherwise, we replace $S_\lambda$ and $S^M_\lambda$ respectively by $S_\lambda/S_\varnothing$ and $S^M_\lambda/S^0_\varnothing$.

In this paper we will discuss the following three properties for $S_\lambda$ or $S^M_\lambda$.
For the explanation, we choose ``the complete symmetric functions'' $H_\ell^M$ ($\ell\in{\mathbb Z}_{\ge 0}$, $M\in{\mathbb Z}$) to be independent indeterminates with the understanding $H_\ell^M=0$ for $\ell<0$ and the normalization $H_0^M=1$, and consider a ring automorphism $\varphi$ generated by $\varphi\big(H_\ell^M\big)=H_\ell^{M+1}$ for all $\ell$ and $M$.

\begin{Definition}The function $S^M_\lambda$ (normalized by $S^0_\varnothing=1$, otherwise $S^M_\lambda/S^0_\varnothing$) is {\it shifted Giambelli} compatible if there exist a set of functions $S'_{(\alpha'\,|\,\beta')}$ labelled by two integers ${\mathbb Z}\times{\mathbb Z}\setminus{\mathbb Z}_{<0}\times{\mathbb Z}_{<0}$, so that the function $S^M_\lambda$ for any Young diagram $\lambda$ is expressed as
\begin{gather}
S^M_\lambda=\det\begin{pmatrix}
\bigl(S'_{(\widetilde\alpha'_i|\beta'_j)}\bigr)
_{\begin{subarray}{c}1\le i\le M\\1\le j\le M+R'\end{subarray}}\\
\bigl(S'_{(\alpha'_i|\beta'_j)}\bigr)
_{\begin{subarray}{c}1\le i\le R'\\1\le j\le M+R'\end{subarray}}
\end{pmatrix},\label{sG}
\end{gather}
when $M\ge 0$ while expressed as
\begin{gather}
S^M_\lambda=\det\begin{pmatrix}
\bigl(S'_{(\alpha'_i|\widetilde\beta'_j)}\bigr)
_{\begin{subarray}{c}1\le i\le|M|+R'\\1\le j\le|M|\end{subarray}}&
\bigl(S'_{(\alpha'_i|\beta'_j)}\bigr)
_{\begin{subarray}{c}1\le i\le|M|+R'\\1\le j\le R'\end{subarray}}
\end{pmatrix},\label{sGnegM}
\end{gather}
when $M=-|M|\le 0$. Here $\alpha'_i$, $\widetilde\alpha'_i$ and $\beta'_j$, $\widetilde\beta'_j$ are respectively (auxiliary) arm lengths and (auxiliary) leg lengths in the $M$-shifted Frobenius notation defined in~\eqref{armleg}--\eqref{auxnegM}.
\end{Definition}

\begin{Definition}The function $S_\lambda$ (normalized by $S_\varnothing=1$, otherwise $S_\lambda/S_\varnothing$) is {\it Jacobi--Trudi} compatible\footnote{In \cite{Macdonald} the Jacobi--Trudi identity is defined also by extending the matrix size in \eqref{JT} to be greater than $L$. Since this extension does not make any differences essentially, we restrict the matrix size to be $L$. The Jacobi--Trudi identity defined with the automorphism $\varphi$ plays an important role in the study of integrable models and is also called the quantum Jacobi--Trudi identity (see \cite{AKLTZ,HL,JR, KNS,KOS,NNSY,NN,T} for examples).} if there exist a set of functions $H_\ell^M$ labelled by two integers $\ell\in{\mathbb Z}_{\ge 0}$, $M\in{\mathbb Z}$ and normalized by $H_0^M=1$ with an automorphism $\varphi\big(H_\ell^M\big)=H_\ell^{M+1}$ so that the function is expressed by
\begin{gather}
S_\lambda=\det\bigl(\varphi^{j-1}H_{\lambda_i-i+j}\bigr)
_{\begin{subarray}{c}1\le i\le L\\1\le j\le L\end{subarray}},\label{JT}
\end{gather}
with $H_\ell=H_\ell^{M=0}$.
\end{Definition}

\begin{Definition}The function $S_\lambda$ (normalized by $S_\varnothing=1$, otherwise $S_\lambda/S_\varnothing$) is {\it Giambelli} compatible if it satisfies
\begin{gather*}
S_{\lambda}=\det\bigl(S_{(\alpha_i\,|\,\beta_j)}\bigr)
_{\begin{subarray}{c}1\le i\le R\\1\le j\le R\end{subarray}},
\end{gather*}
where $\alpha_i$, $\beta_j$ are the arm length and the leg length in the original Frobenius notation defined in~\eqref{frobenius}, $\lambda=(\alpha_1,\alpha_2,\dots,\alpha_R\,|\,\beta_1,\beta_2,\dots,\beta_R)$.
\end{Definition}

Several remarks follow. Note that for the case of $R'=0$ where $|M|$ is large enough so that the shifted diagonal does not intersect with the Young diagram, the block of $S'_{(\alpha'\,|\,\beta')}$ is missing and the shifted Giambelli compatibility resembles the Weyl formula. Also note that $S'_{(\alpha'\,|\,\beta')}$ and~$H_\ell$ appearing in the definitions coincide with $S'_{(\alpha'\,|\,\beta')}=S^0_{(\alpha'=\alpha\,|\,\beta'=\beta)}$ and $H_\ell=S_{[\ell]}$ respectively as obtained from an appropriate choice of~$\lambda$ and~$M$. The Jacobi--Trudi compatibility is utilized in~\cite{M9} to define an ultimate variation (known as the ninth variation) of the Schur function in an arbitrary representation from the complete symmetric functions $H_\ell=S_{[\ell]}$, containing many generalizations by specific choices of~$H^M_\ell$. The defining equation~\eqref{JT} can alternatively be expressed as
\begin{gather*}
S^0_\lambda=\det\bigl(H^{j-1}_{\lambda_i-i+j}\bigr)_{\begin{subarray}{c}1\le i\le L\\1\le j\le L\end{subarray}}, 
\end{gather*}
where $0$ in $S^0_\lambda$ indicates the construction from $H^{M=0}_\ell$. Then, the following two propositions relating these properties hold.

\begin{Proposition}If the function $S_\lambda$ is Jacobi--Trudi compatible, $S_\lambda$ is Giambelli compatible.
\end{Proposition}

\begin{proof}See \cite[Section I.3, Example 21]{Macdonald}.
\end{proof}

\begin{Corollary}If the function $S_\lambda$ is Jacobi--Trudi compatible, $S^M_\lambda$ defined by
\begin{gather}
S^M_\lambda=\det\bigl(H^{M+j-1}_{\lambda_i-i+j}\bigr)
_{\begin{subarray}{c}1\le i\le L\\1\le j\le L\end{subarray}},\label{JTM}
\end{gather}
is Giambelli compatible.
\end{Corollary}

\begin{proof}Since $S^M_\lambda=\det\bigl(\varphi^{j-1}H^M_{\lambda_i-i+j}\bigr)_{\begin{subarray}{c}1\le i\le L\\1\le j\le L\end{subarray}}$ constructed from $H^M_\ell$ is also Jacobi--Trudi compatible, the Giambelli compatibility follows from the above proposition.
\end{proof}

\begin{Proposition}For the shifted Giambelli compatible function $S^M_\lambda$, the Giambelli compatibility holds for a fixed integer $M$,
\begin{gather}
S^M_{\lambda}=\det\bigl(S^M_{(\alpha_i\,|\,\beta_j)}\bigr)_{\begin{subarray}{c}1\le i\le R\\1\le j\le R\end{subarray}}.\label{giambelli}
\end{gather}
\end{Proposition}

\begin{proof}The proof for $M\ge 0$ is given in \cite{GC}. We do not prove for $M<0$ here since it follows from our theorem discussed later.
\end{proof}

\looseness=-1 Hence, it is natural to ask whether the Jacobi--Trudi compatibility is satisfied for the general shifted Giambelli compatible function $S^M_\lambda$, and if yes, what the automorphism is. After our adjustment of the notation, it is not difficult to imagine that the automorphism is that shifting~$M$ by $+1$,
\begin{gather}
\varphi\colon \ M\to M+1.\label{shift}
\end{gather}
In this paper we shall prove the following main theorem. Note that of course the identification of the automorphism~\eqref{shift} is non-trivial since $S^M_\lambda$ in the definition of the shifted Giambelli compatibility \eqref{sG}, \eqref{sGnegM} is totally unrelated to that defined from the automorphism in~\eqref{JTM}. After proving this theorem, the proposition of \cite{GC} \eqref{giambelli} follows directly by combining this theorem with the above corollary. Note also that our theorem is a natural generalization of the theorem (Jacobi--Trudi formula) in \cite{NNSY} which corresponds to the case of $M\ge 0$ and $R'=0$. We believe that our definition of the shifted Giambelli compatibility has clarified how to interpolate the condition between the cases of positive $M$ and negative $M$ with the absolute values $|M|$ being large enough. Here for simplicity, in the theorem, we drop the normalizations $S^M_\varnothing$, $H^{M+j-1}_{0}$ appearing in $S^M_\lambda/S^M_\varnothing$, $H^{M+j-1}_{\lambda_i-i+j}/H^{M+j-1}_{0}$ and assume tacitly that these normalizations are non-vanishing. The dual N\"agelsbach--Kostka compatibility $S^M_\lambda=\det\bigl(E^{M-j+1}_{(\lambda^\text{T})_i-i+j}\bigr) _{\begin{subarray}{c}1\le i\le A\\1\le j\le A\end{subarray}}$ with $A=\card\big\{i\,|\,\big(\lambda^\text{T}\big)_i>0\big\}$ can be proved in the parallel manner by the involution $\lambda\leftrightarrow\lambda^\text{T}$, $M\leftrightarrow-M$.

\begin{Theorem}For the shifted Giambelli compatible function $S^M_\lambda$, the Jacobi--Trudi compatibility holds for a fixed integer~$M$,
\begin{gather}
S^M_\lambda=\det\bigl(H^{M+j-1}_{\lambda_i-i+j}\bigr)_{\begin{subarray}{c}1\le i\le L\\1\le j\le L\end{subarray}}.\label{thm}
\end{gather}
\end{Theorem}

\subsection{Physical application}

The main application we have in mind is the so-called ABJM matrix model,
\begin{gather*}
\langle s_\lambda\rangle_k(N_1,N_2) =\frac{(-1)^{\frac{1}{2}N_1(N_1-1)+\frac{1}{2}N_2(N_2-1)}}{N_1!N_2!}
\int\frac{{\rm d}^{N_1}\mu}{(2\pi)^{N_1}}\frac{{\rm d}^{N_2}\nu}{(2\pi)^{N_2}}e^{\frac{ik}{4\pi}\big(\sum\limits_{m=1}^{N_1}\mu_m^2-\sum\limits_{n=1}^{N_2}\nu_n^2\big)}\nonumber\\
\hphantom{\langle s_\lambda\rangle_k(N_1,N_2) =}{} \times\frac{\prod\limits_{m<m'}^{N_1}\big(2\sinh\frac{\mu_m-\mu_{m'}}{2}\big)^2
\prod\limits_{n<n'}^{N_2}\big(2\sinh\frac{\nu_n-\nu_{n'}}{2}\big)^2}
{\prod\limits_{m=1}^{N_1}\prod\limits_{n=1}^{N_2}\big(2\cosh\frac{\mu_m-\nu_n}{2}\big)^2}s_\lambda\big(e^\mu|e^\nu\big).
\end{gather*}
As reviewed in \cite{PTEP} the hyperbolic functions can be regarded as the trigonometric deformation of the U$(N_1\,|\,N_2)$ invariant measure
\begin{gather}
\frac{\prod\limits_{m<m'}^{N_1}(x_m-x_{m'})^2\prod\limits_{n<n'}^{N_2}(y_n-y_{n'})^2}{\prod\limits_{m=1}^{N_1}\prod\limits_{n=1}^{N_2}(x_m+y_n)^2},\label{invmea}
\end{gather}
with the substitution $x_m=e^{\pm\mu_m}$, $y_n=e^{\pm\nu_n}$ and $s_\lambda(x\,|\,y)$ is the super Schur polynomial, the character of the super unitary group ${\rm U}(N_1\,|\,N_2)$, which shares many interesting properties with the original Schur polynomial such as the Giambelli identity and the Jacobi--Trudi identity~\cite{PT}. Though we introduce the ABJM matrix model as a definition, this matrix model was originally obtained by computing the one-point function of the half-BPS Wilson loop in the ABJM theory \cite{ABJ, ABJM,HLLLP2}, which is the ${\mathcal N}=6$ superconformal Chern--Simons theory describing the worldvolume of the M2-branes. The theory has gauge group ${\rm U}(N_1)_k\times {\rm U}(N_2)_{-k}$ (with the subscripts denoting the Chern--Simons levels) and two pairs of bifundamental matters. Due to the localization techniques, the one-point function originally defined by the infinite-dimensional path integral in the supersymmetric field theory reduces to a finite-dimensional matrix integration~\cite{DT, KWY}.

Without loss of generality hereafter we shall choose $k>0$. After moving to the grand canonical ensemble\footnote{The definition of the vacuum expectation value in the grand canonical ensemble by matching the power of $z$ with one of the arguments is motivated from the study of rank deformations in \cite{MNN} with more complicated gauge groups \cite{MN1,MN3}.} with $M=N_2-N_1$,
\begin{gather*}
\langle s_\lambda\rangle^\text{GC}_{k,M}(z)=\sum_{N=\max(0,-M)}^\infty z^{N}\langle s_\lambda\rangle_k(N,N+M),
\end{gather*}
it was proved in \cite{shifted} that $\langle s_\lambda\rangle^\text{GC}_{k,M}(z)$ is shifted Giambelli compatible \eqref{sG}, \eqref{sGnegM} for $\alpha'_1+\beta'_1+1<k/2$ following the proposal of the Fermi gas formalism~\cite{MP}.
The main steps of the proof in \cite{shifted} are as follows.
\begin{itemize}\itemsep=0pt
\item We rewrite the invariant measure \eqref{invmea} as a product of two determinants, each of which is a combination of the Cauchy determinant and the Vandermonde determinant.
\item We utilize a determinant formula for the super Schur polynomial $s_\lambda(x\,|\,y)$ \cite{MVdJ} which expresses $s_\lambda(x\,|\,y)$ as a ratio of two determinants, with the denominator coinciding with the above Cauchy--Vandermonde determinant.
\item We apply a continuous version of the Cauchy--Binet determinant to combine all the determinants into a single one \cite{shifted} (see also~\cite{PTEP,GC} for a combinatorial viewpoint).
\end{itemize}
Using the property of the shifted Giambelli compatibility, we were able to prove the Giambelli compatibility \cite{HHMO,GC}.

The question whether the Jacobi--Trudi compatibility is satisfied for the shifted Giambelli compatible function, which we have raised in the previous subsection, can be rephrased in the current setup by asking whether the quantum one-point functions of the half-BPS Wilson loop in the ABJM theory satisfy the Jacobi--Trudi identity. In other words, after answering it positively, we find that the Macdonald's ninth variation of the Schur function \cite{M9} defined abstractly by the Jacobi--Trudi compatibility is realized by the path integral.

By applying our theorem in the previous subsection to the ABJM matrix model with the identification of $S_\lambda^M$ as $\langle s_\lambda\rangle^\text{GC}_{k,M}(z)$, or more correctly after the normalization,
\begin{gather*}
S_\lambda^M=\frac{\langle s_\lambda\rangle^\text{GC}_{k,M}(z)} {\langle 1\rangle^\text{GC}_{k,M}(z)},\qquad H_\ell^M=\frac{\langle h_\ell\rangle^\text{GC}_{k,M}(z)} {\langle 1\rangle^\text{GC}_{k,M}(z)},
\end{gather*}
we find that
\begin{gather*}
\frac{\langle s_\lambda\rangle^\text{GC}_{k,M}(z)}{\langle 1\rangle^\text{GC}_{k,M}(z)}
=\det\left(\frac{\langle h_{\lambda_i-i+j}\rangle^\text{GC}_{k,M+j-1}(z)}{\langle 1\rangle^\text{GC}_{k,M+j-1}(z)}\right)
_{\begin{subarray}{c}1\le i\le L\\1\le j\le L\end{subarray}}.
\end{gather*}
As was pointed out in \cite{GC}, the theorem also applies to other matrix models including the Gaussian matrix model and the Chern--Simons matrix model in the canonical ensemble and many cousins of the super Chern--Simons matrix models in the grand canonical ensemble.

Note that although in the original setup $\varphi$ is an automorphism, mapping from a mathematical object to itself bijectively, we apply it under the restriction $\alpha'_1+\beta'_1+1<k/2$, which is essential due to the convergence of the integrations. For this reason our formula is only valid within a~certain range of the Young diagram and we shall take $k$ to be large enough so that the parameters satisfy the bounds.

\section{Proof}

In this section we shall prove the Jacobi--Trudi compatibility out of the shifted Giambelli compatibility \eqref{thm}. We start with $M\ge 0$ since all of the ingredients in the Jacobi--Trudi compati\-bi\-li\-ty~\eqref{thm} have the shift $M+j-1\ge 0$ and the shifted Giambelli compatibility is given in~\eqref{sG}. As in the proof of the Giambelli compatibility \cite{GC} we first consider the simpler case of $R'=0$ and then turn to the general case of $R'\ne 0$. After completing the proof for $M\ge 0$ we can turn to the proof for $M\le 0$. An elegant proof for the case of $M\ge 0$ and $R'=0$ was already given in~\cite{NNSY}. Here we provide a rather pedagogical proof.

\subsection[Case of $M\ge 0$ and $R'=0$]{Case of $\boldsymbol{M\ge 0}$ and $\boldsymbol{R'=0}$}

Our first task would be to translate the relation into an explicit determinant formula. As in the proof of the Giambelli compatibility \cite{GC}, due to the setup $M\ge 0$ and $R'=0$, the arm lengths consist only of the auxiliary ones and the relation is simplified.

We consider the Young diagram $\lambda=(\,|\,\beta'_1,\beta'_2,\dots,\beta'_M)$. Since $\beta'_1=M+L-1$ we define $n_i$ ($n_i>n_{i-1}$) as the complementary set of the leg lengths\footnote{We label $\{n_i\}_{i=1}^L$ inversely because of the relation to the arm lengths \eqref{lambdaarm}.}
\begin{gather}
(\beta'_1,\beta'_2,\dots,\beta'_M)\sqcup(n_L,n_{L-1},\dots,n_1)=(M+L-1,M+L-2,\dots,0).
\label{complement}
\end{gather}
Since the leg length $\{\beta'_l\}_{l=1}^M$ is the distance from the shifted diagonal to the horizontal segment of the boundary of the Young diagram, for the boundary to be successive, the complement $\{n_i\}_{i=1}^L$ is the distance to the vertical segment. This means that the Young diagram can be expressed as $\lambda=[\lambda_1,\lambda_2,\dots,\lambda_L]$ with
\begin{gather}
\lambda_i=M+i-1-n_{i}.\label{lambdaarm}
\end{gather}
Then, the $(M+j-1)$-shifted Frobenius notation for the horizontal Young diagram $[\lambda_i-i+j]$ contains the leg lengths from $M+j-1$ to $0$ except for $n_{i}$, since $(M+j-1)-(\lambda_i-i+j)=n_i$.
After substituting the shifted Giambelli expression \eqref{sG} with the normalization $S^M_\lambda/S^M_\varnothing$ into the Jacobi--Trudi expression \eqref{JT}, we find that the identity we want to prove is given by
\begin{gather*}
 \det\bigl(S'_{(k-M-1\,|\,\beta'_l)}\bigr)_{\begin{subarray}{c}1\le k\le M\\1\le l\le M\end{subarray}}
\prod_{j=2}^{L}\det\bigl(S'_{(k-M-j\,|\,M+j-1-l)}\bigr)
_{\begin{subarray}{c}1\le k\le M+j-1\\1\le l\le M+j-1\end{subarray}}
\nonumber\\
\qquad{} =\det\Big(\det\bigl(S'_{(k-M-j\,|\,M+j-l)}\bigr)
_{\begin{subarray}{c}1\le k\le M+j-1\\1\le l\le M+j,\, l\ne M+j-n_{i}\end{subarray}}\Big)
_{\begin{subarray}{c}1\le i\le L\\1\le j\le L\end{subarray}}.
\end{gather*}
Note that the product on the left-hand side comes from the normalization of the Schur functions.

We shall prove this determinant formula for any value of $S'$. On one hand, since the arm lengths are given trivially for $R'=0$ as we have explained at the beginning of this subsection, we wish to drop them from our notation for simplicity. On the other hand, the matrix size varies for different determinants and we cannot simply forget about the arm lengths. For this reason we introduce a set of infinite-dimensional vectors labelled by $j\in{\mathbb Z}_{\ge 0}$ with the infinite components labelled by $i\in{\mathbb Z}_{>0}$,
\begin{gather*}
(\z_j)_{i}=S'_{(-i\,|\,j)},
\end{gather*}
and define a (finite-dimensional) determinant $| \ |$ for $N$ arrays of infinite-dimensional vectors~$\w_j$ as
\begin{gather}
|\w_1\w_2\cdots\w_N|=\det\bigl((\w_{j})_{N+1-i}\bigr) _{\begin{subarray}{c}1\le i\le N\\1\le j\le N\end{subarray}},\label{findet}
\end{gather}
by choosing the last $N$ components from each vector $\w_j=(\dots,(\w_j)_2,(\w_j)_1)^\text{T}$. Then, the identity can be given as follows.

\begin{Lemma}For any set of infinite-dimensional vectors $(\z_j)_{i}$ labelled by $j\in{\mathbb Z}_{\ge 0}$ and $i\in{\mathbb Z}_{>0}$, the determinant formula
\begin{gather}
|\z_{\beta'_1}\z_{\beta'_2}{\cdots}\z_{\beta'_M}| \prod_{j=2}^{L}|\z_{M{+}j{-}2}\z_{M{+}j{-}3}{\cdots}\z_0| =\det\bigl(|\z_{M{+}j{-}1}\z_{M{+}j{-}2}{\cdots}\widecheck\z_{n_{i}}{\cdots}\z_0|\bigr) _{\begin{subarray}{c}1\le i\le L\\1\le j\le L\end{subarray}},\!\!\!\!\!\!\!\label{Lemma}
\end{gather}
holds. Here $\{n_i\}_{i=1}^L$ is a complementary set of $\{\beta'_l\}_{l=1}^M$ defined in~\eqref{complement} and the determinant $| \ |$ for arrays of infinite-dimensional vectors is defined in~\eqref{findet} by choosing as many components as the arrays from the end of the vectors. Also, $\widecheck\z$ stands for removal of the column $\z$ supplemented with the understanding $|\z_{M+j-1}\z_{M+j-2}\cdots\widecheck\z_{n}\cdots\z_0|=0$ for $n\ge M+j$.
\end{Lemma}

We shall prove this equation by induction of $L$. ($M$ can be general.) This equation holds trivially for $L=1$. (The product $\prod\limits_{j=2}^{L}$ on the left-hand side should be interpreted as $1$ since it comes from the normalization.) When this expression holds for $L=\overline L$, for $L=\overline L+1$ by the Laplace expansion along the first column the right-hand side of~\eqref{Lemma} is given by
\begin{gather}
\det\bigl(|\z_{M+j-1}\z_{M+j-2}\cdots\widecheck\z_{n_{i}}\cdots\z_0|\bigr) _{\begin{subarray}{c}1\le i\le\overline L+1\\1\le j\le\overline L+1\end{subarray}}
\label{laplace}\\
\qquad{}=\sum_{k=1}^{\overline L+1}(-1)^{k-1} |\z_{M}\z_{M-1}\cdots\widecheck\z_{n_{k}}\cdots\z_0|
\det\bigl(|\z_{M+j-1}\z_{M+j-2}\cdots\widecheck\z_{n_{i}}\cdots\z_0|\bigr)_{\begin{subarray}{c}1\le i\le\overline L+1,i\ne k\\2\le j\le\overline L+1\end{subarray}}.\nonumber
\end{gather}
If we rename
\begin{gather*}
\overline M=M+1,\qquad
\overline j=j-1,\qquad \overline i=\begin{cases}i,& 1\le i\le k-1 ,\\i-1, & k+1\le i\le\overline L+1,\end{cases}\qquad
\overline n_{\overline i}=n_{i},
\end{gather*}
or in other words, $(\overline n_{\overline L},\overline n_{\overline L-1},\dots,\overline n_1)=(n_{\overline L+1},n_{\overline L},\dots,\widecheck n_{k},\dots,n_1)$, the complicated minor determinant in \eqref{laplace} can be computed by the assumption of the induction
\begin{gather}
 \det\bigl(|\z_{M+j-1}\z_{M+j-2}\cdots\widecheck\z_{n_{i}}\cdots\z_0|\bigr)
_{\begin{subarray}{c}1\le i\le\overline L+1,i\ne k\\2\le j\le\overline L+1\end{subarray}}
=\det\bigl(|\z_{\overline M+\overline j-1}\z_{\overline M+\overline j-2}\cdots
\widecheck\z_{\overline n_{\overline i}}\cdots\z_0|\bigr)
_{\begin{subarray}{c}1\le\overline i\le\overline L\\1\le\overline j\le\overline L\end{subarray}}
\nonumber\\ \qquad {} =|\z_{\overline\beta'_1}\z_{\overline\beta'_2}\cdots\z_{\overline\beta'_{\overline M}}|
\prod_{\overline j=2}^{\overline L}|\z_{\overline M+\overline j-2}\z_{\overline M+\overline j-3}\cdots\z_0|
\nonumber\\ \qquad{} =|\z_{\beta'_1}\z_{\beta'_2}\cdots\z_{n_{k}}\cdots\z_{\beta'_{M}}|
\prod_{j=3}^{\overline L+1}|\z_{M+j-2}\z_{M+j-3}\cdots\z_0|.
\end{gather}
Here in the second equation we have used the assumption of the induction.
Since $\{\overline\beta'_{\overline l}\}_{\overline l=1}^{\overline M}$ is the complement of $\{\overline n_{\overline i}\}_{\overline i=1}^{\overline L}$ with respect to the set of non-negative integers less than $\overline M+\overline L=M+L$, $n_{k}$ which is missing in $\{\overline n_{\overline i}\}_{\overline i=1}^{\overline L}$ has to be included in $\{\overline\beta'_{\overline l}\}_{\overline l=1}^{\overline M}$.
Then, the right-hand side of \eqref{laplace} becomes
\begin{gather*}
 \det\bigl(|\z_{M+j-1}\z_{M+j-2}\cdots\widecheck\z_{n_{i}}\cdots\z_0|\bigr)
_{\begin{subarray}{c}1\le i\le\overline L+1\\1\le j\le\overline L+1\end{subarray}}\\
\quad =\sum_{k=1}^{\overline L+1}(-1)^{k-1}|\z_{M}\z_{M-1}\cdots\widecheck\z_{n_{k}}\cdots\z_0|
|\z_{\beta'_1}\z_{\beta'_2}\cdots\z_{n_{k}}\cdots\z_{\beta'_M}| \prod_{j=3}^{\overline L+1}|\z_{M+j-2}\z_{M+j-3}\cdots\z_0|.
\end{gather*}
To complete the proof we need to show
\begin{gather*}
 \sum_{k=1}^{\overline L+1}(-1)^{k-1}|\z_{M}\z_{M-1}\cdots\widecheck\z_{n_{k}}\cdots\z_0|
|\z_{\beta'_1}\z_{\beta'_2}\cdots\z_{n_{k}}\cdots\z_{\beta'_M}|
\nonumber\\ \qquad {} =|\z_{\beta'_1}\z_{\beta'_2}\cdots\z_{\beta'_M}||\z_{M}\z_{M-1}\cdots\z_0|,
\end{gather*}
where the first determinant and the second one of each term are of dimension $M$ and $M+1$ respectively. Due to this reason, we add one auxiliary vector~$\z$ in front of the first determinant to prove the corresponding identity with all of the determinants of dimension~$M+1$, where the determinant $| \ |$ defined in~\eqref{findet} reduces to the usual one. This is true because of the Laplace expansion along the first $M+1$ rows,
\begin{gather*}
 \left|\begin{matrix}
\z&\!\!\!\z_{\beta'_1}\z_{\beta'_2}\cdots\z_{\beta'_M}&\!\!\!\z_{M}\z_{M-1}\cdots\z_0\\
\z&\!\!\!{\bm 0}\;\;{\bm 0}\;\;\cdots\;\;{\bm 0}&\!\!\!\z_{M}\z_{M-1}\cdots\z_0\end{matrix}
\right|
=|\z\z_{\beta'_1}\z_{\beta'_2}\cdots\z_{\beta'_M}|
|\z_{M}\z_{M-1}\cdots\z_0|\nonumber\\
 \qquad {} +\sum_{k=1}^{\overline L+1}(-1)^{k}|\z\z_{M}\z_{M-1}\cdots\widecheck\z_{n_{k}}\cdots\z_0|
|\z_{\beta'_1}\z_{\beta'_2}\cdots\z_{n_{k}}\cdots\z_{\beta'_M}|,
\end{gather*}
where the matrix of the double size on the left-hand side is vanishing because of the rectangular blocks of zero components. Finally we set $\z=(1,0,\dots,0)^\text{T}\in{\mathbb R}^{M+1}$ to obtain the desired equation.

\subsection[Case of $M\ge 0$ and $R'\ne 0$]{Case of $\boldsymbol{M\ge 0}$ and $\boldsymbol{R'\ne 0}$}

Let us turn to the $R'\ne 0$ case, $\lambda=(\alpha'_1,\alpha'_2,\dots,\alpha'_{R'}\,|\,\beta'_1,\beta'_2,\dots,\beta'_{M+R'})$. As in the case of $R'=0$ we first define a complementary set of leg lengths. Since $\beta'_1=M+L-1$, we can separate the continuous integer sets into two disjoint sets
\begin{gather*}
(M+L-1,M+L-2,\dots,0,-1,-2,\dots,-R') \\
\qquad{} =(\beta'_1,\beta'_2,\dots,\beta'_{M+R'})\sqcup(n_L,n_{L-1},\dots,n_1),
\end{gather*}
with $n_i>n_{i-1}$. Since all of the leg lengths $\{\beta'_l\}_{l=1}^{M+R'}$ are non-negative, $(-1,-2,\dots,-R')$ has to fall into the last $R'$ components of $\{n_i\}_{i=1}^L$
\begin{gather}
n_{i}=-R'-1+i,\qquad 1\le i\le R' .\label{nlarge}
\end{gather}
As in the previous case of $R'=0$, the Young diagram can be reexpressed as $\lambda=[\lambda_1,\lambda_2,\dots,\lambda_L]$ with
\begin{gather*}
\lambda_i=\begin{cases}
M+i+\alpha'_i,&\text{for }1\le i\le R',\\
M+i-1-n_{i},&\text{for }R'+1\le i\le L.
\end{cases}
\end{gather*}
Then, as before, the $(M+j-1)$-shifted Frobenius notation for $[\lambda_{i}-i+j]$ ($R'+1\le i\le L$) in the lower block contains the leg lengths from $M+j-1$ to $0$ except for $n_{i}$. For this case the determinant formula we need to prove is
\begin{gather}
 \det\begin{pmatrix}
\bigl(S'_{(k-M-1|\beta'_l)}\bigr)_{\begin{subarray}{c}1\le k\le M\\1\le l\le M+R'\end{subarray}}\\
\bigl(S'_{(\alpha'_k|\beta'_l)}\bigr)_{\begin{subarray}{c}1\le k\le R'\\1\le l\le M+R'\end{subarray}}
\end{pmatrix}
\prod_{j=2}^{L}\det\bigl(S'_{(k-M-j|M+j-1-l)}\bigr)
_{\begin{subarray}{c}1\le k\le M+j-1\\1\le l\le M+j-1\end{subarray}}\nonumber\\ \qquad
{} =\det\begin{pmatrix}
\left(\det\begin{pmatrix}
\bigl(S'_{(k-M-j|M+j-l)}\bigr)
_{\begin{subarray}{c}1\le k\le M+j-1\\1\le l\le M+j\end{subarray}}\\
\bigl(S'_{(\alpha'_i|M+j-l)}\bigr)
_{\begin{subarray}{c}1\le l\le M+j\end{subarray}}
\end{pmatrix}\right)
_{\begin{subarray}{c}1\le i\le R'\\1\le j\le L\end{subarray}}\\
\left(\det\bigl(S'_{(k-M-j|M+j-l)}\bigr)
_{\begin{subarray}{c}1\le k\le M+j-1\\1\le l\le M+j,l\ne M+j-n_{i}\end{subarray}}\right)
_{\begin{subarray}{c}R'+1\le i\le L\\1\le j\le L\end{subarray}}
\end{pmatrix}. \label{Rne0}
\end{gather}

We can prove this determinant formula by applying the previous lemma tentatively to another Young diagram and specializing it afterwards. Namely, for the Young diagram
\begin{gather*}
\big[(M+R')^{R'},\lambda_{R'+1},\lambda_{R'+2},\dots,\lambda_{L}\big]=(\,|\,\beta'_1+R',\beta'_2+R',\dots,\beta'_{M+R'}+R'),
\end{gather*}
the determinant formula we have proved in \eqref{Lemma} for the shift $M+R'$ is
\begin{gather*}
 |\z_{\beta'_1+R'}\z_{\beta'_2+R'}\cdots\z_{\beta'_{M+R'}+R'}|
\prod_{j=2}^{L}|\z_{M+j-2+R'}\z_{M+j-3+R'}\cdots\z_0| \\
\qquad{} =\det\bigl(|\z_{M+j-1+R'}\z_{M+j-2+R'}\cdots\widecheck\z_{n_{i}+R'}\cdots\z_0|\bigr)
_{\begin{subarray}{c}1\le i\le L\\1\le j\le L\end{subarray}},
\end{gather*}
where $\{n_i+R'\}_{i=1}^L$ is the complement of $\{\beta'_l+R'\}_{l=1}^{M+R'}$ with respect to the set of non-negative integers less than $M+L+R'$. We are free to rename the vectors by subtracting all of the labels by~$R'$
\begin{gather*}
 |\z_{\beta'_1}\z_{\beta'_2}\cdots\z_{\beta'_{M+R'}}|
\prod_{j=2}^{L}|\z_{M+j-2}\z_{M+j-3}\cdots\z_0\z_{-1}\cdots\z_{-R'}| \\
 \qquad {}=\det\begin{pmatrix}
\bigl(|\z_{M+j-1}\z_{M+j-2}\cdots\z_0\z_{-1}\cdots\widecheck\z_{n_{i}}\cdots\z_{-R'}|\bigr)
_{\begin{subarray}{c}1\le i\le R'\\1\le j\le L\end{subarray}}\\
\bigl(|\z_{M+j-1}\z_{M+j-2}\cdots\widecheck\z_{n_{i}}\cdots\z_0\z_{-1}\cdots\z_{-R'}|\bigr)
_{\begin{subarray}{c}R'+1\le i\le L\\1\le j\le L\end{subarray}}
\end{pmatrix},
\end{gather*}
and identifying the components of the infinite-dimensional vectors $(\z_{j})_i$ as
\begin{gather}
(\z_{j\ge 0})_i=\begin{cases}S'_{(R'-i\,|\,j)},&\text{for }R'+1\le i,\\
S'_{(\alpha'_{R'+1-i}\,|\,j)},&\text{for }1\le i\le R',
\end{cases}
\qquad
(\z_{-j<0})_i=\begin{cases}1,&\text{for }i=R'+1-j,\\
0,&\text{otherwise}.
\end{cases}\label{special}
\end{gather}
Then, the determinants on the left-hand side and those in the lower block on the right-hand side reduce to those we want to prove in \eqref{Rne0}, since the insertion of $\z_{-i}$ means that the row with the arm length~$\alpha'_i$ is eliminated. For the upper block, on the other hand, due to the removal of the column $\z_{n_{i}}=\z_{-R'-1+i}$ \eqref{nlarge} from the determinant, we do not eliminate the arm length~$\alpha'_{R'+1-i}$.
Hence the determinant becomes
\begin{gather}
 |\z_{M+j-1}\z_{M+j-2}\cdots\z_0\z_{-1}\cdots\widecheck\z_{n_{i}}\cdots\z_{-R'}|
\nonumber\\ \qquad{} =(-1)^{R'-i}
\det\begin{pmatrix}
\bigl(S'_{(k-M-j|M+j-l)}\bigr)_{\begin{subarray}{c}1\le k\le M+j-1\\1\le l\le M+j\end{subarray}}\\
\bigl(S'_{(\alpha'_{R'+1-i}|M+j-l)}\bigr)
_{\begin{subarray}{c}1\le l\le M+j\end{subarray}}
\end{pmatrix}, \label{lowerblock}
\end{gather}
with the specification \eqref{special}, where the arm lengths reduce to
\begin{gather*}
\big(\dots,-2,-1,\alpha'_{R'+1-i}\big)^\text{T},
\end{gather*}
and the sign $(-1)^{R'-i}$ appears due to the change of rows. Note that compared with \eqref{Rne0}, the arm lengths in \eqref{lowerblock} appear in the reverse order. To change back to the original order in~\eqref{Rne0} in the determinant, we need to take care of the signs $(-1)^{\sum\limits_{i=1}^{R'-1}i}$. Finally, we find that the total signs cancel
\begin{gather*}
(-1)^{\sum\limits_{i=1}^{R'-1}i}(-1)^{\sum\limits_{i=1}^{R'}(R'-i)}=1,
\end{gather*}
and the determinant reduces to the desired expression \eqref{Rne0}.

\subsection[Case of $M\le 0$]{Case of $\boldsymbol{M\le 0}$}

After completing the proof for $M\ge 0$, let us turn to $M\le 0$. We shall prove the state\-ment~\eqref{thm} by proving it for $M\ge-\overline M$ in the induction of $\overline M$. ($L$ can be general.)

First of all, from the previous two subsections, the statement is true for $M\ge 0$. Suppose this is true for $M\ge-\overline M$, let us prove the case of $M=-(\overline M+1)$,
\begin{gather}
S_\lambda^{-(\overline M+1)} =\det\bigl(H_{\lambda_i-i+j}^{-(\overline M+1)+j-1}\bigr)
_{\begin{subarray}{c}1\le i\le L\\1\le j\le L\end{subarray}}, \label{M+1}
\end{gather}
for the Young diagram
\begin{gather*}
\lambda=[\lambda_1,\lambda_2,\dots,\lambda_L] =\big(\alpha'_1,\alpha'_2,\dots,\alpha'_{\overline M+1+R'}\,|\,\beta'_1,\beta'_2,\dots,\beta'_{R'}\big),
\end{gather*}
with the identification $\alpha'_i=\lambda_i-i+\overline M+1$ for $1\le i\le\overline M+1+R'$. We first Laplace-expand the right-hand side of~\eqref{M+1} along the first column
\begin{gather}
\det\bigl(H_{\lambda_i-i+j}^{-(\overline M+1)+j-1}\bigr)
_{\begin{subarray}{c}1\le i\le L\\1\le j\le L\end{subarray}}
=\sum_{k=1}^L(-1)^{k-1}H_{\lambda_k-k+1}^{-(\overline M+1)}
\det\bigl(H_{\lambda_i-i+j}^{-(\overline M+1)+j-1}\bigr)
_{\begin{subarray}{c}1\le i\le L,i\ne k\\2\le j\le L\end{subarray}}. \label{Laplaceexp}
\end{gather}
If we rename the various variables by
\begin{gather*}
\overline j=j-1,\qquad \overline i=\begin{cases}i,&\text{for }1\le i\le k-1,\\
i-1,&\text{for }k+1\le i\le L,\end{cases} \qquad
\overline\lambda_{\overline i}
=\begin{cases}\lambda_{i}+1,&\text{for }1\le i\le k-1,\\
\lambda_{i},&\text{for }k+1\le i\le L,\end{cases}
\end{gather*}
which implies $\overline\lambda=[\lambda_1+1,\lambda_2+1,\dots,\lambda_{k-1}+1,\lambda_{k+1},\lambda_{k+2},\dots,\lambda_{L}]$, since $\overline\lambda_{\overline i}-\overline i+\overline j=\lambda_i-i+j$, the determinant on the right-hand side of \eqref{Laplaceexp} can be expressed as
\begin{gather*}
 \det\bigl(H_{\lambda_i-i+j}^{-(\overline M+1)+j-1}\bigr)
_{\begin{subarray}{c}1\le i\le L,i\ne k\\2\le j\le L\end{subarray}}
=\det\bigl(H_{\overline\lambda_{\overline i}-\overline i+\overline j}
^{-\overline M+\overline j-1}\bigr)
_{\begin{subarray}{c}1\le\overline i\le L-1\\1\le\overline j\le L-1\end{subarray}}\\
\qquad{} =S_{\overline\lambda}^{-\overline M}
=S_{[\lambda_1+1,\dots,\lambda_{k-1}+1,\lambda_{k+1},\dots,\lambda_{L}]}^{-\overline M},
\end{gather*}
where in the second equation we have used the assumption of the induction. Note that although the summation is taken up to $L$, since $H^M_\ell=0$ for $\ell<0$ we can truncate the summation within $\lambda_k-k+1\ge 0$.
Finally, the identity we want to prove using the shifted Giambelli compatibility~is
\begin{gather}
S_\lambda^{-(\overline M+1)} =\sum_{k=1}^{L}(-1)^{k-1}H_{\lambda_k-k+1}^{-(\overline M+1)}
S_{[\lambda_1+1,\dots,\lambda_{k-1}+1,\lambda_{k+1},\dots,\lambda_{L}]}^{-\overline M}.\label{trunc}
\end{gather}
Since the shifted Giambelli compatibility is given by the $M$-shifted Frobenius notation, we first express all of the Young diagrams in this notation
\begin{gather}
 \det\begin{pmatrix}\bigl(S'_{(\alpha'_i|-\overline M-2+j)}\bigr)
_{\begin{subarray}{c}1\le i\le\overline M+1+R'\\1\le j\le\overline M+1\end{subarray}}&
\bigl(S'_{(\alpha'_i|\beta'_j)}\bigr)
_{\begin{subarray}{c}1\le i\le\overline M+1+R'\\1\le j\le R'\end{subarray}}
\end{pmatrix}
\det\begin{pmatrix}
\bigl(S'_{(\overline M-i|-\overline M-1+j)}\bigr)
_{\begin{subarray}{c}1\le i\le\overline M\\1\le j\le\overline M\end{subarray}}
\end{pmatrix}
\nonumber\\
 \quad{} =\sum_{k=1}^{L}(-1)^{k-1}\det\begin{pmatrix}
\bigl(S'_{(\alpha'_k|-\overline M-2+j)}\bigr)
_{\begin{subarray}{c}1\le j\le\overline M+1\end{subarray}}\\
\bigl(S'_{(\overline M-i|-\overline M-2+j)}\bigr)
_{\begin{subarray}{c}1\le i\le\overline M\\1\le j\le\overline M+1\end{subarray}}
\end{pmatrix}
\nonumber\\
 \qquad{} \times\det\begin{pmatrix}
\bigl(S'_{(\alpha'_i|-\overline M-1+j)}\bigr)
_{\begin{subarray}{c}1\le i\le\overline M+1+R',i\ne k\\1\le j\le\overline M\end{subarray}}&
\bigl(S'_{(\alpha'_i|\beta'_j)}\bigr)
_{\begin{subarray}{c}1\le i\le\overline M+1+R',i\ne k\\1\le j\le R'\end{subarray}}
\end{pmatrix}.\label{SS=SS}
\end{gather}
This can be obtained by the expansion of the determinant
\begin{gather}
\det\begin{pmatrix}
S'_{(\alpha'_i|-\overline M-1)}
&(S'_{(\alpha'_i|-\overline M-2+j)})
_{\begin{subarray}{c}2\le j\le\overline M+1\end{subarray}}
&(S'_{(\alpha'_i|\beta'_j)})
_{\begin{subarray}{c}1\le j\le R'\end{subarray}}
&(S'_{(\alpha'_i|-\overline M-1+j)})
_{\begin{subarray}{c}1\le j\le\overline M\end{subarray}}\\
S'_{(\overline M-i|-\overline M-1)}
&{\bm 0}_{\begin{subarray}{c}2\le j\le\overline M+1\end{subarray}}
&{\bm 0}_{\begin{subarray}{c}1\le j\le R'\end{subarray}}
&(S'_{(\overline M-i|-\overline M-1+j)})
_{\begin{subarray}{c}1\le j\le\overline M\end{subarray}}
\end{pmatrix}
\nonumber\\
\quad{} =(-1)^{(\overline M+R')\times\overline M}\nonumber\\
\qquad{} \times
\det\begin{pmatrix}
(S'_{(\alpha'_i|-\overline M-2+j)})
_{\begin{subarray}{c}1\le j\le\overline M+1\end{subarray}}
&(S'_{(\alpha'_i|-\overline M-1+j)})
_{\begin{subarray}{c}1\le j\le\overline M\end{subarray}}
&(S'_{(\alpha'_i|\beta'_j)})
_{\begin{subarray}{c}1\le j\le R'\end{subarray}}\\
(S'_{(\overline M-i|-\overline M-2+j)})
_{\begin{subarray}{c}1\le j\le\overline M+1\end{subarray}}
&{\bm 0}_{\begin{subarray}{c}1\le j\le\overline M\end{subarray}}
&{\bm 0}_{\begin{subarray}{c}1\le j\le R'\end{subarray}}
\end{pmatrix},\label{SSSS}
\end{gather}
where \looseness=1 for the upper blocks the column indices run over $1\le i\le\overline M+1+R'$ while for the lower blocks the column indices run over $1\le i\le\overline M$. This equation is obtained by exchan\-ging the $\overline M+R'$ column vectors in the second and third blocks with the $\overline M$ column vectors in the fourth block. On one hand, for the left-hand side, if we eliminate the fourth upper block $\bigl(S'_{(\alpha'_i|-\overline M-1+j)}\bigr)_{\begin{subarray}{c}1\le j\le\overline M\end{subarray}}$ from the elementary column operations, the determinant reduces to the left-hand side of~\eqref{SS=SS}. On the other hand if we Laplace-expand the right-hand side of~\eqref{SSSS} along the first $\overline M+1$ columns, we obtain the right-hand side of~\eqref{SS=SS}. In the Laplace expansion, the remaining determinant would be vanishing unless we choose the last $\overline M$ row vectors $\bigl(S'_{(\overline M-i|-\overline M-2+j)}\bigr) _{\begin{subarray}{c}1\le j\le\overline M+1\end{subarray}}$ in the first lower block. To avoid choosing the same row vectors resulting in the vanishing determinant, the remaining row vector has to satisfy $\alpha'_i\ge M$, which means $\alpha_i-i+1\ge 0$, the same truncation in the summation as in~\eqref{trunc}. Hence, finally the same statement is true for $M\ge-(\overline M+1)$. This completes the proof by induction.

\section{Conclusions}

In this paper we prove that the Jacobi--Trudi identity with a shift of the background holds for the one-point functions of the half-BPS Wilson loop in the ABJM matrix model. As noted in~\cite{GC} for the proof of the Giambelli identity, there are many cousins of this matrix model and the proof of the Jacobi--Trudi identity is generalized directly to these matrix models as well. For example we can replace the super unitary group by the super orthosymplectic group studied in~\cite{H,MN5,MS1,MS2,O} which is originally coming from the ${\mathcal N}=5$ super Chern--Simons theories~\cite{ABJ, HLLLP2}. Also we can apply the results here to the matrix models \cite{HHO,MNN,MN1,MN2,MN3,MNY} coming from the ${\mathcal N}=4$ super Chern--Simons theories. These applications are understood directly by the main theorem summarized in this paper.

Our main motivation of proving the Jacobi--Trudi identity is as follows. It was conjectured that the partition function of the ABJM theory is described by the free energy of the closed topological string theory \cite{DMP,FHM,HMMO,HMO3,HMO2,MP} and the one-point function of the half-BPS Wilson loop in the ABJM theory is described by the free energy of the open topological string theory~\cite{HHMO}. Despite the well-established conjecture, it has been difficult to prove it. The Jacobi--Trudi identity proved in this paper, the Giambelli identity proved in~\cite{GC} and the open-closed duality~\cite{HO,KM} strongly suggest the structure of integrable hierarchy \cite{AKLTZ,KNS,Sato} on the matrix model side. Combining with the integrable structure on the topological string side \cite{ADKMV,BGT,GHM2}, we hope that this structure is helpful in proving the conjecture.

\subsection*{Acknowledgements}
We are especially grateful to Soichi Okada for raising the question discussed in this paper clearly, Yasuhiko Yamada and Sintarou Yanagida for many valuable discussions and instructive comments.
We would also like to thank Heng-Yu Chen, Balog Janos, Naotaka Kubo, Satsuki Matsuno, Masatoshi Noumi, Junji Suzuki, Akihiro Tsuchiya and Marcus Werner for valuable discussions.
The work of S.M.\ is supported by JSPS Grant-in-Aid for Scientific Research (C) \#26400245.

\pdfbookmark[1]{References}{ref}
\LastPageEnding

\end{document}